\documentclass[11pt]{article}
\usepackage[T1]{fontenc}  
\usepackage[a4paper, margin=1in]{geometry}
\usepackage{amsmath,amssymb,amsthm,amsfonts,cite,xcolor,xspace,graphicx,xifthen,xparse,bm,thmtools,thm-restate,appendix}
\usepackage{enumitem}
\usepackage{bold-extra}
\usepackage[hang]{footmisc}
\usepackage{wrapfig}
\usepackage{pgf, tikz}
\usetikzlibrary{patterns,shapes.arrows,shapes.misc}
\setitemize{noitemsep,topsep=3pt,parsep=3pt,partopsep=3pt}
\usepackage{subcaption}  
\definecolor{darkgreen}{rgb}{0,0.5,0}
\definecolor{darkblue}{rgb}{0,0,0.8}
\usepackage{hyperref}
\hypersetup{
		hypertexnames=false,
    unicode=false,          
    colorlinks=true,        
    linkcolor=darkblue,          
    citecolor=darkgreen,        
    filecolor=magenta,      
    urlcolor=cyan           
}
\usepackage{algorithm}
\usepackage[noend]{algpseudocode}
\usepackage[nameinlink,capitalize]{cleveref}

\newtheorem{theorem}{Theorem}[section]
\newtheorem{lemma}[theorem]{Lemma}
\newtheorem{corollary}[theorem]{Corollary}
\newtheorem{observation}{Observation}

\theoremstyle{definition}
\newtheorem{definition}[theorem]{Definition}
\theoremstyle{remark}
\newtheorem{remark}[theorem]{Remark}

\Crefname{definition}{Definition}{Definitions}
\Crefname{lemma}{Lemma}{Lemmas}
\Crefname{claim}{claim}{Claims}
\Crefname{corollary}{Corollary}{Corollaries}
\Crefname{remark}{Remark}{Remarks}
\Crefname{observation}{Observation}{Observations}

\makeatletter
\newcounter{algorithmicH}
\let\oldalgorithmic\algorithmic
\renewcommand{\algorithmic}{%
  \stepcounter{algorithmicH}
  \oldalgorithmic}
\renewcommand{\theHALG@line}{ALG@line.\thealgorithmicH.\arabic{ALG@line}}
\makeatother

\algnewcommand\algorithmicvariable{\textbf{Variables:}}
\algnewcommand\Var{\item[\algorithmicvariable]}

\NewDocumentCommand{\ceil}{s O{} m}{%
    \IfBooleanTF{#1}
    {\left\lceil#3\right\rceil} 
    {#2\lceil#3#2\rceil} 
}

\NewDocumentCommand{\floor}{s O{} m}{%
    \IfBooleanTF{#1}
    {\left\lfloor#3\right\rfloor} 
    {#2\lfloor#3#2\rfloor} 
}
\newcommand{\abs}[1]{\left| #1 \right|}

\newcommand{\id}{\text{ID}}
\newcommand{\myparagraph}[1]{\par\noindent \textit{ #1}}
\newcommand{\rank}{\ensuremath{ \ell}}
\newcommand{\degree}{\ensuremath{ \Delta}}
\newcommand{\congestmodel}{\ensuremath{\mathsf{CONGEST}}\xspace}
\newcommand{\localmodel}{\ensuremath{\mathsf{LOCAL}}\xspace}
\newcommand{\logstar}{\log^{*}}
\newcommand{\poly}{\mathrm{poly}}
\newcommand{\polylog}{\poly\log}
\newcommand{\dist}{\text{dist}}
\newcommand{\diversity}{\mathcal{D}}
\newcommand{\activeNodes}{\mathcal{A}}

\begin{document}

\begin{flushleft}

\vspace*{2cm}
{\huge\bf Deterministic Distributed Ruling Sets of \\ Line Graphs\footnote{A preliminary version of this paper appeared in
    the 25$\mathrm{th}$  International Colloquium on Structural Information and Communication Complexity (SIROCCO 2018) \cite{SIROCCO2018}.}\par}
\vspace{2cm}

\newcommand{\auth}[3]{\textbf{#1}$\,\,\,\cdot\,\,\,$#2$\,\,\,\cdot\,\,\,$#3\par\medskip}

\auth{Fabian Kuhn\footnote{Supported by ERC Grant No.\ 336495 (ACDC).}}
{University of Freiburg}
{kuhn@cs.uni-freiburg.de}
\auth{Yannic Maus\footnotemark[3]}
{University of Freiburg}
{yannic.maus@cs.uni-freiburg.de}
\auth{Simon Weidner}
{University of Freiburg}
{simon.weidner@cs.uni-freiburg.de}

\end{flushleft}

\vspace{1cm}

\paragraph{Abstract.}
An $(\alpha,\beta)$-ruling set of a graph $G=(V,E)$ is a set
  $R\subseteq V$ such that for any node $v\in V$ there is a node
  $u\in R$ in distance at most $\beta$ from $v$ and such that any two
  nodes in $R$ are at distance at least $\alpha$ from each other. The
  concept of ruling sets can naturally be extended to edges, i.e., a
  subset $F\subseteq E$ is an \emph{$(\alpha,\beta)$-ruling edge set}
  of a graph $G=(V,E)$ if the corresponding nodes form an
  $(\alpha,\beta)$-ruling set in the line graph of $G$.  This paper
  presents a simple deterministic, distributed algorithm, in the
  \congestmodel~ model, for computing $(2,2)$-ruling edge sets in
  $O(\logstar n)$ rounds. Furthermore, we extend the algorithm to
  compute ruling sets of graphs with bounded \emph{diversity}. Roughly
  speaking, the diversity of a graph is the maximum number of maximal
  cliques a vertex belongs to.  We devise $(2,O(\diversity))$-ruling
  sets on graphs with diversity $\diversity$ in
  $O(\diversity+\logstar n)$ rounds. This also implies a fast,
  deterministic $(2,O(\rank))$-ruling edge set algorithm for
  hypergraphs with rank at most \rank.

  Furthermore, we provide a ruling set algorithm for general graphs
  that for any $B\geq 2$ computes an
  $\big(\alpha, \alpha \ceil{\log_B n}\big)$-ruling set in
  $O(\alpha \cdot B \cdot \log_B n)$ rounds in the \congestmodel
  model. The algorithm can be modified to compute a
  $\big(2, \beta \big)$-ruling set in
  $O(\beta \degree^{2/\beta} + \logstar n)$ rounds in the
  \congestmodel model, which matches the currently best known such algorithm
  in the more general \localmodel model.
 \clearpage

\section{Introduction, Motivation \& Related Work}

This paper presents fast and simple deterministic distributed algorithms, in the \congestmodel\ model, for computing ruling sets of graphs, line graphs, line graphs of hypergraphs, and graphs of bounded diversity as introduced in \cite{BEM17}.

\paragraph{The \congestmodel{} Model of Distributed Computing~\cite{peleg2000distributed}.} The graph is abstracted as an $n$-node network $G=(V, E)$ with maximum degree at most $\Delta$. Each node is assumed to have a unique $O(\log n)$-bit ID. Communication happens in synchronous rounds. Per round, each node can send one message of at most $O(\log n)$ bits to each of its neighbors and perform (unbounded) local computations\footnote{All our algorithms only use local computations that require at most polynomial time.}. At the end, each node should know its own part of the output, e.g., whether it belongs to the ruling set or not. The time complexity of an algorithm is the number of rounds it requires to terminate.

\paragraph{Ruling Sets.} A \emph{$\big(\alpha, \beta \big)$-ruling set} of a graph $G=(V,E)$ is a subset $R\subseteq V$ of the nodes such that any two nodes in $R$ are at distance at least $\alpha$ in $G$ and for every node $v\in V\setminus R$, there is a node in $R$ within distance $\beta$ \cite{awerbuch1989network}. That is, $R$ is an independent set in $G^{\alpha-1}$, where $G^r$ denotes the graph with node set $V$ and where two nodes $u,v$ are connected by an edge if $d_G(u,v)\leq r$.  Typically, $\alpha$ is called the \emph{independence parameter} and $r$ the \emph{domination parameter} of the ruling set $R$.  If $\alpha=2$, one often also simply calls $R$ a \emph{$\beta$-ruling set}. The concept of ruling sets can naturally be extended to edges, i.e., a subset $F\subseteq E$ is an \emph{$(\alpha,\beta)$-ruling edge set} of a graph $G=(V,E)$ (or a hypergraph $H=(V,E)$) if the corresponding nodes form an $(\alpha,\beta)$-ruling set in the line graph of $G$ (or $H$). In the present paper, we concentrate on deterministic algorithms for computing ruling sets in the \congestmodel model. We will specifically see that edge ruling sets of graphs and low-rank hypergraphs can be computed particularly efficiently.

\paragraph{The Relevance of Ruling Sets.} 
The distributed computation of ruling sets is a simple and clean symmetry breaking problem. In particular, ruling sets are a generalization of maximal independent sets (MIS), arguably one of the most central and best studied distributed symmetry breaking problems. A $(2,1)$-ruling set is an MIS of $G$ and more generally, a $(r+1,r)$-ruling set is an MIS of $G^r$. For $\beta\geq 1$, a $(2,\beta)$-ruling set of $G$ is therefore a strict relaxation of an MIS of $G$, where the problem becomes weaker with larger values of $\beta$. The parameter $\beta$ thus allows to study a trade-off between the strength of the symmetry breaking requirement and the complexity of computing a ruling set.

Ruling sets have been introduced by Awerbuch, Goldberg, Luby, and Plotkin as a building block to efficiently construct a so-called network decomposition (a partition of a graph into clusters of small diameter together with a coloring of the cluster graph with a small number of colors) \cite{awerbuch1989network}. Since then, ruling sets have been used as a powerful tool in various distributed graph algorithms. Computing ruling sets can often replace computing the more stringent and harder to obtain maximal independent sets. Ruling sets have for example been used in order to compute network decompositions \cite{awerbuch1989network,panconesi95}, graph colorings \cite{panconesi95delta}, maximal independent sets \cite{kuhn05fast}, or shortest paths \cite{henzinger2016deterministic}. Ruling sets are also used as a subroutine to obtain the state-of-the-art randomized distributed algorithms for many of the classic distributed graph problems, such as distributed coloring \cite{chang2017optimal,harris16coloring}, maximal independent set \cite{ghaffari2016improved}, or maximal matching \cite{barenboim2016locality}. These algorithms are based on the so-called graph shattering technique, which was originally introduced by Beck in \cite{beck1991algorithmic}. Using an efficient randomized algorithm, the problem is solved on most of the graph such that the only unsolved remaining parts are components of size at most $\poly(\Delta\cdot\log n)$, where $\Delta$ is the maximum degree of $G$. Using existing ruling set algorithms, one can then further reduce the problems on the remaining components to problems on graphs of size $\polylog n$.

\paragraph{Previous Work on Distributed Ruling Set Algorithms.} 
As mentioned before the first appearance of ruling sets was in the work of Awerbuch, Goldberg, Luby, and Plotkin \cite{awerbuch1989network}, who provided a deterministic distributed algorithm to compute an $(\alpha, O(\alpha \log n))$-ruling set in $O(\alpha \log n)$ rounds.  Their algorithm uses the bit representation of the node IDs to recursively compute ruling sets. For each of the $O(\log n)$ bits of the IDs, the nodes are divided into two parts according to the value of the current bit and ruling sets are computed recursively for the two parts. The two recursively computed ruling sets $S_0$ and $S_1$ are merged to a single ruling set by keeping all nodes in $S_0$ and all nodes in $S_1$ that do not have an $S_0$-node in their $(\alpha-1)$-neighborhood. Note that the algorithm loses an additive $\alpha$ in the domination for each of the $O(\log n)$ recursion levels.

Schneider and Wattenhofer refine the ideas of \cite{awerbuch1989network} to deterministically compute $(2,\beta)$-ruling sets in time $O(\beta \degree^{2/\beta} + \log^* n)$ \cite{schneider2013symmetry}. At the cost of an increased running time, they use a larger branching factor than Awerbuch et al.\ to decrease the recursive depth and thus the domination parameter. Further, for small values of $\beta$, the best known deterministic algorithm requires time $2^{O(\sqrt{\log n})}$, even for $\beta=1$. It is based on first computing a network decomposition using the algorithm of \cite{panconesi95} and to then use this decomposition to compute the ruling set. All these algorithms work in the \localmodel model, where the size of messages is unbounded. In \cite{henzinger2016deterministic}, Henzinger, Krinninger, Nanongkai sketch how the algorithm of \cite{awerbuch1989network} can be adapted to compute a $\big(\alpha, O(\alpha \log n) \big)$-ruling set in $O(\alpha \log n)$ rounds in the \congestmodel model.  Any $(2,\beta)$-ruling set algorithm applied to $G^{\alpha-1}$ implies a $(\alpha,(\alpha-1)\beta)$-ruling set algorithm on $G$, e.g., \cite{schneider2013symmetry} can be used to compute $\big(\alpha, O(\alpha^2\cdot \beta) \big)$-ruling sets in time $O(\alpha^2 \beta \degree^{2 / \beta} + \alpha \cdot \log^* n)$. However, the black box simulation of an algorithm on $G^{\alpha-1}$ heavily relies on the \localmodel model.  

In contrast to the few deterministic ruling set algorithms there are many randomized algorithms for the problem.  
In particular, there is a long history of efficient randomized algorithms for computing an MIS. The famous algorithms by Luby and Alon, Babi, and Itai allow to compute an MIS (and thus a $(2,1)$-ruling set) in time $O(\log n)$ \cite{alon1986fast,luby1986simple}. In \cite{ghaffari2016improved}, Ghaffari improved the randomized running time of computing an MIS to $O(\log \Delta) + 2^{O(\sqrt{\log\log n})}$. There are also more efficient randomized algorithms that directly target the computation of $(2,\beta)$-ruling sets for $\beta>1$. 
Gfeller and Vicari found an algorithm that finds a $(1, O(\log \log \degree))$-ruling set in time $O(\log \log \degree)$ such that the degree in the graph induced by the ruling set nodes is $O(\log^5 n)$ \cite{gfeller2007randomized}. Together with the ruling set algorithms of \cite{awerbuch1989network,schneider2013symmetry}, the algorithm allows to compute a $(2,O(\log\log n))$-ruling set in time $O(\log\log n)$. In
 \cite{barenboim2016locality}, Barenboim et al.\ used the graph shattering technique to compute $(2, \beta)$-ruling set in time $O(\beta \log^{1 / (\beta -1/2)}) + 2^{O({\sqrt{\log \log n}}})$. This was later improved by Ghaffari to compute $(2,\beta)$-ruling sets in $O(\beta \log^{1 / \beta}) + O(2^{\sqrt{\log \log n}}) $ rounds \cite{ghaffari2016improved}.
Kothapalli and Pemmaraju showed how to compute $\big(2,2\big)$-ruling sets in $O(\log^{3/4} n)$ rounds \cite{kothapalli2012super}. The core idea is a randomized sparsification process that reduces the degree while maintaining some domination property. Afterwards the algorithm of \cite{barenboim2016locality} is applied to the sparsified graph.
The same authors presented a randomized algorithm that computes $(2,\beta)$-ruling sets in time $O(\beta \log^{1 / (\beta -1)} n)$ if $\beta\leq \sqrt{\log\log n}$ and in time $O(\sqrt{\log \log n})$ for arbitrary $\beta$ \cite{bisht2014brief}.  
Pai et al.\ showed how to compute $3$-ruling sets in $O( {\log n / \log \log n})$ rounds and $2$-ruling sets in $O(\log(\degree) \cdot (\log n) ^{1/2+\epsilon} + {\log n \ \epsilon \log \log n})$ rounds in the \congestmodel model \cite{pai2017symmetry}. Further, the work deals with the \emph{message complexity} of ruling set algorithms. They provide a $\Omega(n^2)$ lower bound on the message complexity for MIS if nodes have no knowledge about their neighbors  and in contrast present a $2$-ruling set algorithm that uses $O(n\log^2 n)$ messages and runs in $O(\Delta\log n)$ rounds. 

We are not aware of any work that explicitly studies ruling edge sets. However, there is substantial work on computing maximal matchings and maximal hypergraph matchings, i.e., on computing $(2,1)$-ruling edge sets. While for the MIS problem no polylogarithmic-time deterministic algorithm is known, in \cite{hanckowiak98}, Ha\'n\'ckowiak, Karo\'nski and Panconesi showed that a maximal matching can be computed in $O(\log^7 n)$ rounds deterministically.  They improved the algorithm to $O(\log^4 n)$ rounds in \cite{hanckowiak1999faster}. The current best algorithm is by Fischer and it computes a maximal matching in time $O(\log^2 \degree \cdot \log n)$ \cite{fischer2017improved}. Fischer, Ghaffari and Kuhn have recently shown that maximal matchings can even be computed efficiently in low-rank hypergraphs. For hypergraphs of rank at most $r$ (i.e., every hyperedge consists of at most $r$ nodes), they presented a deterministic algorithm to compute hypergraph maximal matching in $O(\log^{O(log r)}\Delta\cdot \log n)$ \cite{FGK17}. Later in \cite{GHK17}, the dependency on the rank was improved; the paper obtains a runtime of $\Delta^{O(r)} + O(r \logstar n)$ to compute a hypergraph maximal matching.  Furthermore, \cite{barenboim2016locality} contains a randomized algorithm that (combined with \cite{fischer2017improved}) computes a maximal matching in $O(\log \degree + \log^3 \log n)$ rounds.

Finally, we note that the $\Omega(\logstar n)$ lower bound of \cite{linial1992locality} that was designed for coloring and MIS on a ring network also holds for computing ruling sets. On a ring, given a $\beta$-ruling (edge) set, an MIS can be computed in time $O(\beta)$. Maximal matchings have a lower bound of $O\big(\sqrt{\log n/\log \log n}\big)$ rounds \cite{kuhn2010local}.

\paragraph{Contributions.} The ruling set algorithms for general graphs by Awerbuch et al. \cite{awerbuch1989network}  and Schneider et al. \cite{schneider2013symmetry} only works in the \localmodel model. In \cite{henzinger2016deterministic} Henzinger, Krinninger, and Nanongkai sketch a variant of the algorithm that achieves a ruling set of the same quality as Awerbuch et al.\, but that also works in the \congestmodel model. 
In \Cref{sec:simpleruling} we provide a formal analysis and a generalization of the algorithm of \cite{henzinger2016deterministic}. Further, slightly beyond \cite{henzinger2016deterministic}, our simple deterministic distributed algorithm also levels Schneider et al.'s work in the \congestmodel model. 
\begin{restatable}{theorem}{mainSimpleGraph}\label{thm:mainSimpleGraph}
  Let $\alpha$ be a positive integer. For any $B\geq 2$ there exists a deterministic distributed algorithm that computes a $\big(\alpha, (\alpha-1) \lceil\log_B n\rceil\big)$-ruling set of $G$ in $O(\alpha \cdot B \cdot \log_B n)$ rounds in the $\congestmodel$ model.
\end{restatable}
The $n$ in the runtime of \Cref{thm:mainSimpleGraph} stems from the size of the ID-space. To compute a $(2,\beta)$-ruling set it is sufficient to use the colors of a $c\Delta^2$-coloring computed with Linial's algorithm \cite{linial1987distributive} as IDs; setting $B = c\cdot\degree^{2/\beta}$ implies the same trade-off as in \cite{schneider2013symmetry}.

\begin{restatable}{corollary}{mainSimpleGraphDelta}\label{proof:RulingSetArbitrarydomination}
    Let $\beta > 2$ be an integer. There exists a deterministic distributed algorithm that computes a $\big(2, \beta \big)$-ruling set of $G$ in $O(\beta \degree^{2/\beta} + \logstar n)$ rounds in the $\congestmodel$ model.
\end{restatable}
The simple algorithm of \Cref{thm:mainSimpleGraph} begins with a tentative ruling set $S=V$ and iteratively sparsifies $S$ until it is an independent set. In iteration $i$ of $O(\log_B n)$ iterations it removes nodes from $S$ such that $S$ is independent with regard to the $i$'th digit of the $B$-ary representation of the ids---the set $S$ is called \emph{independent with regard to digit $i$} if both endpoints of each edge of $G[S]$ have the same value at the $i$'th digit. Then, if all bits are independent $S$ is an independent set; the node removal in each iteration is such that the domination increases by at most two in each iteration.

All further contributions center around the fast computation of ruling sets in line graphs and their generalizations. 
The main result is the computation of $2$-ruling edge sets in $O(\logstar n)$ rounds.
\begin{restatable}{theorem}{mainRulingEdgeSet}
\label{lem:fast3rulingedge}
  There exists a deterministic distributed algorithm that computes a $2$-ruling edge set of $G$ in $\Theta(\log^*n)$ rounds in the \congestmodel model.
\end{restatable}
The main idea of the algorithm can actually be explained in a few lines: In the first step each node sends a proposal along one of its incident edges; in the second step each node that received a proposal adds exactly one of the edges through which it received a proposal to a set $F$; in the third step nodes compute a matching on the graph induced by the edges in $F$ and add matching edges to the ruling edge set. One can show that the graph that is induced by the edges in $F$ has maximum degree at most two and thus the computation of the matching only takes  $O(\logstar n)$ rounds. The resulting ruling edge set is a $3$-ruling edge set and we use our following result to transform it into a $2$-ruling edge set. 

The \emph{proposal technique} is similar the one of Israeli et al. in \cite{israeli1986fast} for the randomized computation of maximal matchings. For multiple phases they first reduce the maximal degree of the graph using a randomized version of the proposal algorithm. Then they randomly add certain edges from the reduced graph to the matching and remove their adjacent edges. They show that the algorithm removes a constant fraction of the edges in each phase and thus they obtain a maximal matching in $O(\log n)$ rounds. 

\begin{restatable}{theorem}{mainRulingEdgeSetReduction}
\label{proof:kto2edgerule}
    Let $\beta \geq 2$. Any $\beta$-ruling edge set can be transformed into a $2$-ruling edge set in $O(\beta)$ rounds of communication in the \congestmodel model.
\end{restatable}
We emphasize that a further reduction of the domination parameter, i.e., to $1$-ruling edge sets or equivalently to maximal matchings, cannot be done in less than $O\big(\sqrt{\log n/\log \log n}\big)$ rounds due to the lower bound of \cite{kuhn2010local}.

The algorithm from \Cref{lem:fast3rulingedge} can be seen as a $2$-ruling set algorithm on line graphs  and it is significantly faster than any known algorithm to compute $2$-ruling sets on general graphs.  
Our third contribution extends the ideas of the algorithm for line graphs to a much larger class of graphs, i.e., graphs with bounded diversity. For a graph $G$,  a \emph{clique edge cover} $Q$ is  a collection of cliques of $G$ such that every edge (and node) of $G$ is contained in at least one of the cliques. The \emph{diversity} of a pair $(G,Q)$ where $Q$ is a clique edge cover of $G$ is $\diversity$ if any node is contained in at most $\diversity$ distinct cliques of the clique edge cover. The diversity of a graph $G$ is the minimum diversity of all $(G,Q)$ where $Q$ is an arbitrary clique edge cover of $G$. The concept of diversity was introduced in \cite{BEM17}. In the following we always assume that the clique edge cover $Q$ is known by all nodes, i.e., each node knows all the cliques in which it is contained. Note that for many graphs, e.g., for line graphs of graphs or hypergraphs of small rank a clique edge cover with a small diversity is obtained in a single round of communication by taking a clique for each node $v$ consisting of the set of all the edges containing $v$.  Note that in \cite{BEM17}, the definition of diversity is defined by using maximal cliques. However, we do not require maximality in our algorithms.
\begin{restatable}{theorem}{mainDiversity}
\label{proof:diversityrulingsetresult} 
    There exists a algorithm that, given a graph with a clique edge cover of diversity $\diversity$, computes a $(\diversity+4)$-ruling set in $O(\diversity + \log^*n)$~ rounds in the \congestmodel model.
\end{restatable}
Line graphs have diversity two and (non-uniform) hypergraphs with rank \rank~have diversity \rank. The corresponding clique edge covers can be computed in a single round  which implies the following corollary.
\begin{restatable}{corollary}{mainHypergraph} \label{cor:mainHypergraph}
There exists a algorithm that computes  $(\rank+4)$-ruling edge sets in $O(\rank + \log^*n)$ rounds in (non uniform) hypergraphs with rank at most \rank.
\end{restatable}

\paragraph{Outline.} \Cref{sec:rulingedgesetsofsimplegraphs} focuses on ruling edge sets. We believe that it is helpful to read this section to understand the more involved algorithm in  \Cref{sec:boundedDiversity} which extends results to graphs of bounded diversity and line graphs of hypergraphs. In \Cref{sec:simpleruling}, we present the ruling set algorithm for general graphs in the \congestmodel model.


\section{Ruling Edge Sets of Simple Graphs}\label{sec:rulingedgesetsofsimplegraphs}
We provide an algorithm to compute ruling edge sets with an asymptotically optimal runtime. Even though the same asymptotic runtime can be obtained with the more general algorithm in \Cref{sec:boundedDiversity} (the algorithm works for graphs of bounded diversity --- line graphs have diversity two) we believe that this section is simpler and more straightforward. This section also helps to understand the more involved algorithm in \Cref{sec:boundedDiversity}. In \Cref{sec:proposalsimplegraphs} we show how to compute $3$-ruling edge sets in $O(\logstar n)$ rounds and in \Cref{sec:rulingedgesetreduction} we show how any $\beta$-ruling edge set can be transformed into a $2$-ruling edge set in $O(\beta)$ rounds.

The distance $\dist(e,f)$ between two edges $e$ and $f$  is defined as the distance of the corresponding nodes in the line graph. The graph $G = (V,E)$ induced by a set of vertices $U \subseteq V$ is defined as $G[U] = (U, \{\{u,v\} \mid \{u,v\} \in E, u,v \in U\}$ and the graph induced by a set of edges as $F \subseteq E$ is $G[F] = (V,F)$.
We extend the definition of ruling (vertex) sets to \emph{ruling edge sets}.  
\begin{definition}[Ruling Edge Set]
  An \emph{$(\alpha, \beta)$-ruling edge set} $ R \subseteq E$ of a graph $G = (V, E)$ is a subset of edges such that the distance between any two edges in $R$ is at least $\alpha$ and for every edge $e\in E$ there is an edge $f\in R$ with $\dist(e,f)\leq \beta$.
\end{definition}
In the \localmodel model the computation of ruling edge sets is equivalent to computing ruling sets on line graphs. Note that line graphs have many additional properties, e.g., bounded diversity (cf. \Cref{sec:boundedDiversity}),  and not every graph can appear as a line graph of another graph. 
We focus on ruling edge sets with distance two between any two 'selected' edges. The next remark indicates why:  any independence greater than two immediately leads to results for ruling vertex sets.
\begin{remark}\label{proof:rulingedgesetstorulingsets}
    Any $\big(\alpha, \beta \big)$-ruling edge set with $\alpha \geq 2$, $\beta \geq 1$ directly leads to a $\big(\alpha -1, \beta +1 \big)$-ruling set.
		\end{remark}
\begin{proof}
    Let $S$ be a $\big( \alpha, \beta \big)$-ruling edge set.  For each edge $\{v_1,v_2\} \in S$ add one of the nodes, e.g., $v_1$ to the node set $R$. Isolated nodes are also added to $R$.
    \myparagraph{Independence.} Let $v_1$ and $v'_1$ be two nodes in $R$. By construction there are two distinct edges $e=\{v_1,v_2\}$ and $e'=\{v'_1,v'_2\}$ in $S$. 
		As the distance between $e$ and $e'$ is at least $\alpha$  the shortest path $p$ that contains both edges $e$ and $e'$ has  at least $\alpha+1$ edges. Hence the distance between $v_1$ and $v_2$ is at least $\alpha-1$. 
	    \myparagraph{Domination.} Let $v\in V$ be a node with incident edge $e=\{v,w\}$. Then there is an edge $f=\{v_f,w_f\}\in S$ with $\dist(e,f)\leq \beta$. Either $v_f$ or $w_f$ is contained in $R$ and thus there is a node in distance $\beta+1$ to $v$.
\end{proof}

Maximal matchings are $\big(2,1 \big)$-ruling edges sets. These can be computed in $\polylog n$ time with \cite{fischer2017improved} or with a large dependency on the maximum degree and only a $O(\logstar n)$ dependence on $n$.
\begin{corollary}[\cite{panconesi2001some}]\label{lem:matchingwithcolors}
	Maximal matchings in graphs with maximum degree at most \degree~ can be computed in $O(\degree+\log^* n)$ deterministic distributed time.
\end{corollary}

\subsection{Proposal Technique for Simple Graphs}\label{sec:proposalsimplegraphs}
 In the first step of our ruling edge set algorithm we compute, in a constant number of rounds, a subset $F\subseteq E$ of the  edges such that $(1)$ for every edge $e\in E$ there is an edge $f\in F$ such that  the distance between $e$ and $f$ is small and $(2)$ the graph $G[F]$ has small maximum degree. In the second step we apply any (known) ruling edge set algorithm on the edges of $G[F]$, e.g., the algorithm from \Cref{lem:matchingwithcolors}. We call a set $F$ with these properties an \emph{edge-kernel}. 

\begin{definition}[Edge-kernel]\label{def:kernel}
    Let $G = (V,E)$ be a graph. A \emph{$(d, r)$-edge-kernel} $F \subseteq E$ is a subset of edges, so that the degree of the induced graph $G[F]$ is at most $d$ and for every edge $e \in E$ there exists an edge $f \in F$ with $\dist_G(e, f) \leq r$.
\end{definition}
\vspace{-0.7cm}
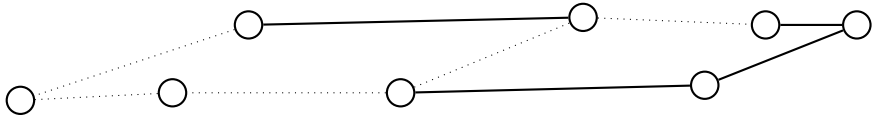
\begin{figure}[H]
  \centering
  \begin{tikzpicture}
    \begin{scope}[every node/.style={circle,thick,draw}]
      \node (A) at (0,0) {};
      \node (B) at (2,-0.9) {};
      \node (C) at (4.4,0.1) {};
      \node (D) at (6.8,0) {};
      \node (E) at (8,0) {};
      \node (F) at (6,-0.8) {};
      \node (G) at (-1,-0.9) {};
      \node (H) at (-3,-1) {};
    \end{scope}

    \begin{scope}[
        every node/.style={fill=white,circle},
            every edge/.style={draw=black,thick}]
          \path [-] (A) edge (C);
          \path [-] (D) edge (E);
          \path [-] (F) edge (E);
          \path [-] (B) edge (F);
    \end{scope}

    \begin{scope}[
        every node/.style={fill=white,circle},
            every edge/.style={draw=black,dotted}]
          \path [-] (B) edge (C);
          \path [-] (C) edge (D);
          \path [-] (B) edge (G);
          \path [-] (G) edge (H);
          \path [-] (A) edge (H);
    \end{scope}
  \end{tikzpicture}
	\vspace{-0.2cm}
  \caption{Non-dotted lines form a $(2, 2)$-edge-kernel.}\label{fig:kernel}
\end{figure}
\vspace{-0.2cm}

The core idea of our algorithm is the proposal technique of the next lemma.
\begin{lemma}\label{lem:proposal}
  There is a deterministic two round \congestmodel algorithm to compute a $(2,2)$-edge-kernel.
\end{lemma}
\textbf{Algorithm: } Each node proposes one of its incident edges and in the next step each node accepts a single of its incident edges that were proposed by \underline{other nodes}. Return the set $F$ of accepted edges.
\begin{algorithm}[H]
  \caption{Proposal Technique}\label{alg:kernelproposal}
  \begin{algorithmic}[1]
    \For{\textbf{each} node $n$ \textbf{in parallel}}
    \State Propose one incident edge to all neighbor nodes
    \State Arbitrarily add one of the edges that are proposed by neighbors to the set $F$
    \EndFor
		\State \Return $F$
  \end{algorithmic}
\end{algorithm}
\begin{proof}
     The algorithm requires two rounds in the \congestmodel model. We claim is that $F$ is a $(d,r)$-edge-kernel with $d,r\leq 2$.

\myparagraph{$d \leq 2$.} Any node $v \in V$ has at most two incident edges in $F$: the edge that $v$ proposed itself if it was accepted by the corresponding neighbor and the edge $v$ accepted. This concludes that $\degree(G[F]) \leq 2$.

\myparagraph{$r \leq 2$.} Consider any edge $e = \{v, u\} \in E$, the vertex $v$ proposed some edge $f = \{v, w\}$ in the first step of the algorithm. Thus $w$ has at least one incident edge that was proposed by a neighbor. Let $g=\{w,w'\}\in F$ be the edge that is accepted by $w$. 
Then the distance between $f$ and $g$ is at most $2$ through the path $e,f,g$. Note that the distance is even smaller if $v$ proposed edge $e$ or if $w$ accepted edge $f$.
\end{proof}

Computing a ruling edge set on an edge-kernel provides a ruling edge set of the original graph  whose domination parameter is the sum of the domination parameters of the edge-kernel and the ruling edge set.

\begin{lemma}\label{proof:kernelrule}
    Let $d, r_1$ and $r_2$ be positive integers. Given an $(d,r_1)$-edge-kernel $F\subseteq E$ of a graph \mbox{$G=(V,E)$} and an $r_2$-ruling edge set algorithm with runtime $T(n,\text{max degree})$, one can compute an \mbox{$(r_1 + r_2)$}-ruling  edge set of $G$ in time $T(n,d)$.
\end{lemma}
\noindent\textbf{Algorithm:} Apply the $r_2$ ruling edge set algorithm on the graph $G[F]$; let $R\subseteq F$ be its output.
\begin{proof} 
    \myparagraph{$R$ is independent on $G$:} Let $e$ and $f$ be two edges in $R\subseteq F$. They are not adjacent in $G[F]$ by the guarantees of the algorithm. Further, if they were adjacent in $G$ then they would, by the definition of the induced graph, also be adjacent in $G[F]$, a contradiction.
		
    \myparagraph{$R$ is $r_1+r_2$ dominating on $G$:} Let $e\in E$ be an arbitrary edge of $G$. Due to the edge-kernel properties there is an edge $f\in F$ with $\dist_G(e,f)\leq r_1$. As $R$ is an $r_2$-ruling edge set in $G[F]$ there is an edge $g\in R$ with $\dist_{G[F]}(f,g)\leq r_2$. This implies that $\dist_G(e,g)\leq \dist_G(e,f)+\dist_G(f,g)\leq r_1+r_2$. 
	\end{proof}

The bottleneck when computing a maximal matching is the maximum degree (cf. \Cref{lem:matchingwithcolors}). An $(d, r)$-edge-kernel reduces the degree to $d$. By first computing an $(2,2)$-edge-kernel and thereafter running our matching algorithm on it we obtain a $\big(2,3 \big)$-ruling edge set.

\mainRulingEdgeSet*
\begin{proof}
  First compute a $(2,2)$-edge-kernel $F$ with \Cref{lem:proposal} in $O(1)$ rounds. Thereafter run the matching algorithm \Cref{lem:matchingwithcolors} on $G[F]$ in $O(\logstar n+2)=O(\logstar n)$ rounds and return the matching. If we formulate these steps in the language of line graphs to apply \Cref{proof:kernelrule}  we obtain that the returned set is a $3$-ruling edge set of $G$. 

Use \Cref{proof:kto2edgerule} to reduce the domination from $3$ to $2$ in $O(1)$ rounds.
The lower bound of $\Omega(\logstar n)$ follows from Linial's lower bound \cite{linial1992locality} as the line graph of the ring forms an isomorphic ring. 
 \end{proof}
The concept of edge-kernels as introduced in \Cref{def:kernel} is not helpful to compute ruling edge sets with independence parameter $\alpha>2$. Given an edge-kernel $F$, we use that we can handle the connected components of $G[F]$ separately as the distance between connected components is at least two. If one was to compute ruling edge sets with independence $\alpha>2$ one had to ensure that the distance between connected components is at least $\alpha$. Note that \Cref{proof:rulingedgesetstorulingsets} implies that such an algorithm would immediately imply an (unknown) algorithm for the computation of a (non trivial) ruling set of $G$.

\subsection{\texorpdfstring{\boldmath From $\beta$-ruling edge sets to $2$-ruling edges sets}{From b-ruling edge sets to 2-ruling edges sets}}\label{sec:rulingedgesetreduction}
In this subsection we show how to decrease the domination parameter of ruling edge sets from $\beta$ to $2$ within $O(\beta)$ rounds. In particular, we show how to transform  $3$-ruling edge sets into $2$-ruling edge sets in $O(1)$ rounds and essentially repeating the algorithm $\beta$ times leads to the result for general $\beta$.

The core idea is adding additional edges to a $\beta$-ruling edge set $R$ to decrease its domination parameter: Let $E_2$ be the set of edges whose shortest distance to an edge in $R$ is two. We carefully select an independent set $I\subseteq E_2$ such that every edge in distance three to $R$ has an edge in distance at most two in $I$. Then $R\cup I$ forms the desired $2$-ruling edge set.
Note that adding all edges with distance two (or any other fixed distance) to $R$ cannot be done without losing independence. Furthermore the induced graph $G[E_2]$  might have degree up to $\Delta$. Thus, to obtain constant runtime, we cannot apply any of the known algorithms with non-constant runtime in a black box fashion to $G[E_2]$. 
We first need one very simple but also very useful observation.

\begin{observation}\label{obs:rulingedgesetstructure}
The distance of any pair of incident edges to the closest edge in an ruling edge set differs at most by one (cf. \Cref{fig:linerule3}). 
\end{observation}

\begin{lemma}\label{lem:3to2edgerule}
  A $3$-ruling edge set can be transferred into a $2$-ruling edge set in $O(1)$ \congestmodel rounds.
\end{lemma}
\noindent\textbf{Algorithm:}
Given a $3$-ruling edge set $R \subseteq E$, we compute a $2$-ruling edge set $R\subseteq S \subseteq E$. 
First, split the edges $E$ into four sets $E = E_0 \cup E_1 \cup \dots \cup E_3$ according to their distance to an edge in $R$. This can be done in three rounds. 
	  Then every node that is adjacent to at least one edge from $E_2$ and at least one edge from $E_3$ selects a single of its incident edges $e_2 \in E_2$ as a \emph{candidate edge}.
    Now, each node with at least one incident candidate edge that also has an incident edge in $E_1$ chooses one of its incident candidate edges and adds it to the set $I$. Finally return the set $S=R\cup I$.
				\vspace{-0.2cm}
  \begin{figure}[H]
    \centering
      \begin{tikzpicture}[scale=0.9]
          \begin{scope}[every node/.style={circle,thick,draw}]
              \node (A) at (0,0) {t};
              \node (A2) at (0, -1.5) {};
              \node (B) at (3,0) {u};
              \node (C) at (6,0) {v};
              \node (D) at (9,0) {};
              \node (D2) at (9, -1.5) {};
          \end{scope}

      \begin{scope}[
          every node/.style={rectangle,above,inner sep=0mm, outer sep=2mm},
              every edge/.style={draw=black,thick}]
            \path [-] (A) edge node[above] (AB) {$1$} (B);
        \path [-, dotted] (A2) edge node[rectangle,below,inner sep=0mm, outer sep=3mm] (A2B) {$ 1 \lor 2 $} (B);
        \path [-] (B) edge node[above] (BC) {$2$} (C);
        \path [-] (C) edge node[above] (CD) {$ 3 $} (D);
        \path [-, dotted] (C) edge node[rectangle,below,inner sep=0mm, outer sep=3mm] (CD2) {$ 2 \lor 3 $} (D2);
      \end{scope}
    \end{tikzpicture}
		\vspace{-0.2cm}
    \caption{Neighborhood of a candidate edge $e=\{u,v\} \in E_2$ proposed by node $v$. Solid edges exist in any graph, dotted edges may exist in any cardinality.}\label{fig:linerule3}
  \end{figure}
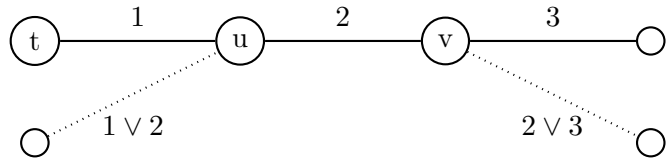
	\vspace{-0.3cm}
\begin{proof}
  \myparagraph{$S$ is independent.} It is helpful to keep the essence of \Cref{obs:rulingedgesetstructure} in mind which implies that the set of nodes that propose an edge cannot be connected by an edge in $E_2$. For contradiction, assume such an edge $e = \{u,v\} \in E_2 $ exists. As both nodes propose an edge they both have incident edges in $E_3$. However, then both nodes do not have an incident edge in $E_1$ which contradicts that $e\in E_2$ (cf. \Cref{fig:linerule3}).
	
	As the set $I$ of added edges is a subset of $E_2$ none of them is adjacent to edge edge in $E_0$. Thus we only need to prove that no two edges in $I$ are adjacent. 
	Assume there are two edges $e = \{u,v\}, f=\{u,w\}$ in $I\subseteq E_2$ that are adjacent. Then $v$ proposed $e$ and $w$ proposed $f$ because if any of the edges would have been proposed by $u$ the other edge could not be proposed at all by our previous observation that proposed nodes cannot be connected by an edge in $E_2$. As the proposals of $e$ and $f$ compete at $u$ and $u$ can only accept one of them not both edges can be contained in $I$.

  \myparagraph{$S$ has domination two.} Edges in $E\setminus E_3$ are still dominated in distance two as $R \subseteq S$. Let $e=\{u,v\}\in E_3$. At least one of its endpoints has an incident edge $f$ in $E_2$ that it proposed. W.l.o.g. assume that $u$ proposed edge $f=\{u,w\}\in E_2$. Either $f$ is accepted by $w$ which implies that $e$ is dominated or $w$ accepted some other edge $g$ which dominates $e$ in distance two. 
  \end{proof}
We use the same idea to improve the domination parameter.
\begin{lemma}\label{lem:ktokminus1edgerule}
    Any $\beta$-ruling edge set with $\beta \geq 3 $ of a simple graph can be transformed into an $\beta-1$-ruling edge set in $O(1)$ rounds of communication in the \congestmodel model. 
\end{lemma}
\begin{proof}
 Let $R$ be the given $\beta$-ruling edge set. First, for each edge $e \in E$ we check whether its $3$-neighborhood contains an edge of $R$ and split the edges $E$ into four sets according to their distance to $R$, i.e.,  $E = E_0 \cup E_1 \dots E_3 \cup E_{\geq 4}$. Let $H=G[E_{\leq 3}]$ be the subgraph induced by the edges $E_{\leq 3} = E_0  \cup  E_1 \cup E_2 \cup E_3$ and apply \Cref{lem:3to2edgerule} to transform $R$ into a $2$-ruling edge set $S$ of $H$ in $O(1)$ rounds. 

  We observe that this $2$-ruling edge set of $H$ is a $(\beta-1)$-ruling edge set of $G$:  In the graph $G$ the shortest path $p={e_\beta, \dots,e_3,e_2,e_1,e_0}$ from any edge $e_\beta$ (with distance $\beta$ to $R$) to $E_0=R$ contains an edge of $E_3$; the indices of the path edges correspond to their distances to $e_0$. As $e_3$ has an edge at distance at most $2$ in $S$, the edge $e_\beta$ has an edge in $S$ at distance at most $\beta-1$.
\end{proof}
\Cref{lem:ktokminus1edgerule} reduces the domination of a $\beta$-ruling edge set in a constant number of rounds, independent from $\beta$. Particularly, the reduction even works in constant time if no node knows how far it is from the closest ruling edge before the algorithm starts.
Iteratively applying \Cref{lem:ktokminus1edgerule} implies the following theorem.
\mainRulingEdgeSetReduction*
\begin{proof}
  Apply \Cref{lem:ktokminus1edgerule} $\beta-2$ times iteratively reducing $\beta$ to $2$ in $O(\beta)$ rounds. 
\end{proof}


\section{Ruling Sets of Bounded Diversity Graphs}\label{sec:boundedDiversity} 
In \Cref{sec:rulingedgesetsofsimplegraphs}, we have seen that the computation of ruling sets on line graphs seems to be much easier than on general graphs. In this section, we identify graph properties that allow us to essentially apply the same algorithm as in \Cref{sec:rulingedgesetsofsimplegraphs} to a much more general class of graphs, in particular to bounded diversity graphs. Bounded diversity was introduced in \cite{BEM17}. Given a graph $G = (V,E)$ and a \emph{clique edge cover} $Q$, i.e., a  set of cliques  (where each clique is a subgraph of $G$) such that any node of $G$ is contained in at least one clique and for any two nodes $u,v \in V$ that are adjacent in $G$ there exists a clique $C\in Q$ in which $u$ and $v$ are adjacent. The \emph{diversity} with respect to the cover $Q$ is the maximal number of cliques a vertex is contained in. The diversity of a graph is the minimum diversity over all clique edge covers. We show that, given such a cover with diversity $\diversity$, we can compute an $O(\diversity)$-ruling set in time $O(\diversity+\logstar n)$. In many cases, e.g., in line graphs and line graphs of hypergraphs clique edge covers with very low diversity can be computed in constant time even in the \congestmodel model. 
\begin{definition}[Diversity]\label{def:diversity}
Given a graph $G=(V, E)$ and a clique edge cover $Q$, the \emph{diversity} of $(G,Q)$ is defined as $\max_{v \in V} \abs{\{C \in Q \mid v \in C  \}}$. The \emph{diversity} of $G$ is the minimum of the diversities of $(G,Q)$ over all clique edge covers $Q$.
\end{definition}
\Cref{def:diversity} is slightly different from the definition in \cite{BEM17} where the cliques are required to be maximal. However, none of our algorithms use this property and going without it might lead to covers with smaller diversity and hence faster runtimes. 
One downside of both definitions is that (so far) algorithms rely on a globally known cover that, in the best case, levels the diversity of the graph. In both models of computation it is not clear that computing such a cover can always be done efficiently. However, in the \localmodel model it is straightforward to compute a cover with diversity $\Delta$: For each node, add all maximal cliques that it is contained in. 
\begin{figure}[H]

  \centering
                    \tikzstyle{very loosely dotted}=[dash pattern=on \pgflinewidth off 5pt]
                    \centering
                    \begin{tikzpicture}[scale=0.8, rotate=-90]
                      \begin{scope}[every node/.style={circle,thick,draw}]
                        \node (A) at (-1.8,0) {$a$};
                            \node (B) at (1.8,0) {$b$};
                            \node (C1) at (0,2.5) {$c_1$};
                            \node (C2) at (0,5) {$c_2$};
                            \node (C3) at (0,7.5) {$c_3$};
                            \node (CL) at (0,12.5) {$c_{\Delta-1}$};
                      \end{scope}

                        \begin{scope}[every edge/.style={draw=black,thick,very loosely dotted}]
                          \draw (C3) edge (CL);
                        \end{scope}
                        \begin{scope}[
                            every node/.style={fill=white,circle},
                            every edge/.style={draw=black,thick}]
                          \draw (A) edge (B);
                          \draw (A) edge (C1);
                          \draw (A) edge (C2);
                          \draw (A) edge (C3);
                          \draw (B) edge (C1);
                          \draw (B) edge (C2);
                          \draw (B) edge (C3);
                          \path (B)  edge (CL);
                          \path (A)  edge (CL);

                        \end{scope}
                    \end{tikzpicture}

                      \caption{In the line graph of the given graph the edge $ab$ is contained in all of the following maximal cliques
                      $\{ab,ac_1,bc_1\}, \{ab,ac_2,bc_2\},\{ab,ac_3,bc_3\}, \dots \{ab,ac_{\Delta -1}, bc_{\Delta -1}\}$ and using them for a clique cover implies a diversity of $\Delta$. However, all edges can be covered with the two cliques $\{ab, ac_1, ac_2, \dots, ac_{\Delta -1}\}$ and $\{ab, ab_1, ab_2, \dots, ab_{\Delta -1}\}$.}
                    \label{fig:DiversityExample}
                \centering
      \label{fig:diversityvsdiversity}
      \end{figure}
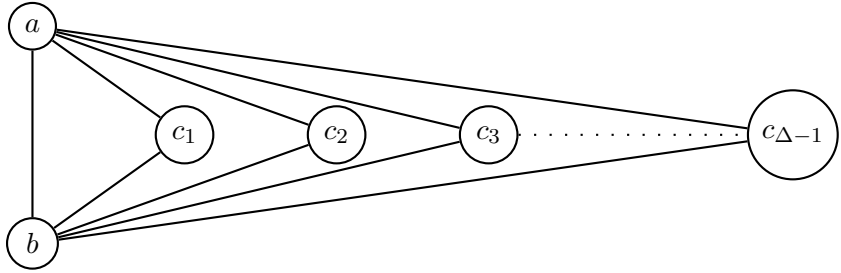
In \Cref{fig:diversityvsdiversity}, we provide an example that shows that this is not necessarily optimal. In the \congestmodel model the seemingly hard problem of triangle detection (see e.g., \cite{LeGall17}) can be reduced to the problem of identifying maximal cliques. Often it is not difficult to compute a clique edge cover with a small diversity, e.g., a cover with diversity two in line graphs can be computed in constant time in the \congestmodel model. For a further discussion of the computability of such covers consult \cite{BEM17}.

A \emph{hypergraph} $H$ is a tuple $(V,E)$ of vertices and edges  and each edge is a set of vertices. The \emph{rank} of a hypergraph is the maximum number of vertices that are contained in an edge.  
One way to define a distributed algorithm on a hypergraph is that, in one round, each vertex $v$ of the hypergraph broadcasts one message on each of its incidents hyperedges $e$ (the messages can be different for different hyperedges but all nodes in the same hyperedge receive the same message from $v$ on that edge) and receives the messages sent by its neighbors.
The diversity of line graphs of hypergraphs is bounded above by the rank of the hypergraph. Simple graphs are (uniform) hypergraphs of rank two and hence have diversity of at most two.

\begin{remark}\label{proof:linegraphofhypergraphdiversity}
The diversity of the line graph of a hypergraph  of rank at most \rank~ is at most \rank~ and a corresponding cover can be computed in constant time.\end{remark}
\begin{proof}Let $H = (V,E)$ be a hypergraph and $L = \mathcal{L}(H)$ its line graph. For each vertex $v$ we define the (constant time computable) clique $C_v = \{e \mid e \in E, v \in e\}$ and $Q=\{C_v \mid v\in V\}$. Then $Q$ is a clique edge cover of $L$ with diversity $\rank$ as an edge $e=\{v_1,\ldots,v_{\rank}\}$ is only contained in the cliques $C_{v_1},\ldots,C_{v_{\rank}}$. \end{proof}



In this section we show how to adapt the proposal technique of \Cref{lem:proposal} in \Cref{sec:rulingedgesetsofsimplegraphs}. Recall that we used the proposal technique to compute $(d,r)$-edge-kernels of a graph. 
 In this section we use vertex-kernels.
\begin{definition}[Vertex-kernel]\label{def:vertexkernel}
    Let $G = (V,E)$ be a graph. A \emph{$(d, r)$-vertex-kernel} $A \subseteq V$ is a subset of nodes, so that the degree of the induced graph $G[A]$ is at most $d$ and for every node $v \in V$ there exists a node $u \in A$ with $\dist_G(v, u) \leq r$.
\end{definition}
First, we  rephrase the proposal technique of \Cref{lem:proposal} directly on the line graph. Each node of the original graph can be identified with a clique in the line graph. The 'proposing an edge' in \Cref{alg:kernelproposal} corresponds to proposing a single node from each such clique. Then, 'accepting a proposed edge' corresponds to accepting one of the proposed nodes in the clique. In \Cref{lem:proposal} we showed that at most two nodes per clique \emph{survive} this process on the line graph. In general many more nodes can survive a single step of this proposal and accepting technique. We repeat the process to sparsify the selected nodes further and further. 

\begin{lemma}\label{proof:diversityproposal}
    There exists an $O(\diversity)$ time algorithm in the \congestmodel model that, given a graph $G=(V,E)$ and a set of cliques $Q$ with diversity $\diversity$,  computes a subset of nodes $\activeNodes$ with the following properties:
    \begin{align*}
        \textit{(small degree)}~ & \text{For all cliques}~ C \in Q : ~ \abs{\activeNodes \cap C} \leq \diversity~,\\
        \textit{(domination)}~ & \text{For all nodes}~ v \in V ~\text{there is a node}~ p \in \activeNodes \text{ with } \dist_G(v,p) \leq \diversity~.
    \end{align*}
	Moreover, $\activeNodes$ is a $(\diversity^2, \diversity)$-vertex-kernel of the graph.
\end{lemma}
\noindent\textbf{Algorithm: }  At the start each node is set \emph{active}. Then, in each of $\diversity$ phases, each clique proposes one of its active nodes  that it has not proposed in any phase before (if such a node exists). Any node that was active before, has not been proposed in the current phase but has a neighbor that is proposed in the current phase is set \emph{inactive}. In the end we return the set of active nodes. Confer \Cref{alg:diversityproposal} for detailed pseudocode.

\begin{algorithm}
    \caption{Proposal Technique for Graphs with Diversity $\diversity$}\label{alg:diversityproposal}
		\begin{tabular}[t]{@{}lll}
            \textbf{Input} &  $Q$ & Clique Cover of $G$ with diversity $\diversity$.\\
		    \textbf{Output} & $\activeNodes$ & $(\diversity^2, \diversity)$-vertex-kernel. \\
            \textbf{Variables} & $\activeNodes$ & set of active nodes, \\
        & $R^C$ & nodes that clique $C$ proposed,	\\
        & $S$ & proposals of the current phase,\\
        & $W$ &  active nodes that are not adjacent to any node of $S$.
	\end{tabular}
    \begin{algorithmic}[1]
		
	\State $R^C = \emptyset$, $S=\emptyset$, $W=\emptyset$
	\State $\activeNodes = V$          \Comment{set all nodes active}
        \For{$\diversity$ \textbf{times}}
        \State $S = \emptyset $ 
        \For{\textbf{each} clique $C \in Q$ \textbf{in parallel}}
            \If{$C \cap \activeNodes \setminus R^C \neq \emptyset$}
                \State propose one node $p^C \in C \cap \activeNodes \setminus R^C$
                \State $R^C = R^{C} \cup \{p^C\}$
                \State $S = S \cup \{p^C\}$
            \EndIf      
        \EndFor
        \State $W = \activeNodes \setminus N_1(S)$\Comment{not-proposed+proposed neighbor $\longrightarrow$ inactive}
        \State $\activeNodes = W \cup S$ \Comment{active nodes for the next phase}
        \EndFor
        \State \Return $\activeNodes$
    \end{algorithmic}
\end{algorithm}
For the correctness of the algorithm we  show that any node in $\activeNodes\cap C$ has been proposed by clique $C$ and as each clique proposes at most one node in each of the $\diversity$ iterations the claim \textit{(small degree)} follows. The second property follows as a node is only set inactive if it has a neighbor that is active in the next phase. 

\begin{proof} 
For $i=1,\ldots,\diversity$ let $\activeNodes_i$ denote the set of nodes that are active at the end of phase $i$, $S_i$  the set of nodes that are proposed in phase $i$, $W_i$  the set of nodes that are active at the end of phase $i$ and do not have a neighbor that is proposed in phase $i$ and $R_i^C$ be the set of nodes that have been proposed by clique $C$ until phase $i$. To prove the lemma we first prove the following property: $(1)$ $S_1\supseteq S_2\supseteq \ldots \supseteq S_{\diversity}~.$ 

 Assume for contradiction, that $v\notin S_j$ and $v\in S_{j+1}$ for some $j<d$. Let $C$ be a clique that proposes $v$ in phase $j+1$. In phase $j+1$ only nodes in $\activeNodes_j$ can be proposed. Thus $v$ is contained in $\activeNodes_j=W_j\cup S_j$. As $v$ is not contained in $S_j$ we deduce that $v\in W_j$, i.e., $v$ does not have a neighbor that is proposed in phase $j$. In particular, $C$ does not propose a neighbor of $v$ in round $j$, i.e., either $C$ proposed $v$ in phase $j$ or $v$ does not propose any node in phase $j$ at all. In both cases $C$ cannot propose $v$ in phase $j+1$, a contradiction.

Fix a clique $C\in Q$. We show that each node in $P \cap C$ has been proposed by $C$. As $C$ proposes at most $\diversity$ nodes the claim \textit{(small degree)} follows. 
   If there is an $i<d$ with $(\activeNodes_i \cap C) \setminus R_i^C=\emptyset$ the claim holds because $\activeNodes\cap C\subseteq \activeNodes_i\cap C$ and clique $C$ already proposed all nodes in $\activeNodes_i\cap C$ in the first $i$ rounds. So assume that $(\activeNodes_{\diversity-1}\cap C) \setminus R_{\diversity-1}^C\neq \emptyset$ and let  $v$ be the node that $C$ proposes in the last phase. All nodes in $C$ that are not proposed in phase $\diversity$ are set inactive as their neighbor $v$ is proposed.  Thus any node in $\activeNodes \cap C$ is a node in $S_{\diversity}$, i.e., any such node is proposed in phase $\diversity$ by some clique and due to Property $(1)$ also in each phase before. Thus any node in $\activeNodes\cap C$ has been proposed by $\diversity$ many cliques. As no clique can propose a node twice and each node is in at most $\diversity$ cliques (including clique $C$) each such node has been proposed by $C$. 
		
    \myparagraph{Domination.}
    At the start every node has an active neighbor (i.e., a neighbor in $\activeNodes$). A node is only set inactive (i.e., removed from $\activeNodes$) in some phase if it has neighbor that is proposed in the phase. Thus the domination distance increases at most by one per phase which proves the claim. 

    \myparagraph{The algorithm runs in $O(\diversity)$ rounds in the \congestmodel model.} 
		In a single phase, removing the non proposed nodes with a proposed neighbor from the set of active nodes can be done in a single round. Thus the runtime is in the order of the number of phases, i.e., it is $O(\diversity)$. 
		
		\myparagraph{$\activeNodes$ is a $(\diversity^2, \diversity)$-vertex-kernel.} Due to diversity $\diversity$ any  $v$ can only be part of $\diversity$ distinct cliques. Due to the \textit{(small degree)} property it has at most $\diversity-1$ neighbors in each clique. As the cliques cover every edge of $G$ this implies that the maximum degree of $G[\activeNodes]$ is upper bounded by $\diversity(\diversity-1)\leq \diversity^2$.
    The domination follows immediately from the second property.
\end{proof}
Analogously to \Cref{proof:kernelrule} one can prove the following lemma.
\begin{lemma}\label{proof:vertexkernelrule}
    Let $d, r_1$ and $r_2$ be positive integers. Given an $(d,r_1)$-vertex-kernel $S\subseteq V$ of a graph \mbox{$G=(V,E)$} and an $r_2$-ruling set algorithm with runtime $T(n,\text{max degree})$ one can compute an \mbox{$(r_1 + r_2)$}-ruling  set of $G$ in time $T(n,d)$.
\end{lemma}
\Cref{proof:diversityproposal,proof:vertexkernelrule} and the ruling set algorithm from \Cref{proof:RulingSetArbitrarydomination} imply the main result of the section.	

\mainDiversity*
 \begin{proof}
   First use \Cref{proof:diversityproposal} to compute a $(\diversity^2,\diversity)$-vertex-kernel $\activeNodes$ of $G$ in $O(\diversity)$ rounds. Then run the  ruling set algorithm from \Cref{proof:RulingSetArbitrarydomination} with $\beta = 4$ on $G[\activeNodes]$ in $O(\diversity + \log^*n)$. With \Cref{proof:vertexkernelrule} this yields a $(\diversity+4)$-ruling set.
\end{proof}

As hypergraphs of rank $\rank$ have diversity $\rank$ and we can efficiently compute a corresponding clique decomposition and we obtain \Cref{cor:mainHypergraph}.
\mainHypergraph*


\section{Ruling Sets of Simple Graphs}\label{sec:simpleruling}

Awerbuch et al. present in \cite{awerbuch1989network} an algorithm that computes $(\alpha, \alpha \log n)$-ruling sets in $O(\alpha \log n)$ rounds in the \localmodel model.  With similar ideas  Schneider, Elkin and Wattenhofer showed in \cite{schneider2013symmetry} how to compute $(2,\beta)$-ruling sets in $O(\beta \degree^{2/\beta} + \log^* n)$ time in the \localmodel model , i.e., they trade in runtime for a better domination. \cite{awerbuch1989network} has been transformed to the \congestmodel model in \cite{henzinger2016deterministic} and we present an algorithm that levels the parameters of Schneider et al.'s algorithm in the \congestmodel model.

An $\big(\alpha, \beta\big)$-ruling set is a subset of the vertices of a graph such that any node of the graph can reach some vertex of the ruling set in $\beta$ steps and the distances of any two vertices of the ruling set is at least $\alpha$. 
\begin{definition}[Ruling Set]
  An \emph{$\big(\alpha, \beta \big)$-ruling set} $ R \subseteq V$ of a graph $G = (V, E)$ is a subset of nodes such that the distance in $G$ between any two nodes in $R$ is at least $\alpha$ and for every node $v\in V$ there exists a node $u\in R$ with $\dist(v,u)\leq \beta$. $\alpha$ is called the \emph{independence} parameter and $\beta$ the \emph{domination} or \emph{ruling} parameter of the ruling set. As usual we write \emph{$\beta$-ruling} for a $\big(2,\beta\big)$-ruling set.
\end{definition}

\textbf{Algorithm:} The algorithm starts with the trivial dominating set $R=V$ and sparsifies it at the cost of increasing the domination until the distance between any two nodes in the set is at least $\alpha$.  A set $R$ is called \emph{$\alpha$-independent in bit $i$} if two nodes whose binary representation differs in the $i$'s bit  have distance at least $\alpha$. 
To obtain full independence we iterate through the bits of the binary representation of the ids and produce independence in all bits. 
Let $i$ be the position of the bit for which we want to produce independence: Let $U_b[i]$ be the nodes of $G$ that have bit $b$ at the $i$'s position of their bit string. We partition the set $R$ into the nodes $R_0=R\cap U_0[i]$ that have a zero at the $i$'s bit and the nodes $R_1=R\cap U_1[i]$ that have a one at the $i$'s bit. All nodes of $R_0$ remain in $R$ and we remove all nodes of $R_1$ from $R$ that have a node in distance less than $\alpha$ in $R_0$. 
This step can be executed for all nodes in parallel in the \congestmodel model: Start multiple BFS searches (with $\alpha-1$ hops) at the nodes in $R_0$ and whenever  two or more breadth-first searches are conflicting, only a single search, i.e., one with most hops left, is forwarded. 

In \Cref{thm:finalsimple} we show that one step of the algorithm increases the domination parameter by at most $\alpha$. We provide an example of this step in \Cref{fig:examplerulingset}. The runtime and the final domination is $O(\alpha\cdot \log_2 n)$ because the bit strings have length $O(\log_2 n)$. In our algorithm and for some $B\geq 2$ we use the $B$-ary representation of the ids instead of the binary representation. Then the final domination is $O(\alpha\log_B n)$ if we can still increase the domination parameter by at most $\alpha$ for each digit: To handle a single bit we partition $R$ into $B$ sets $R_0=R\cap U_0[i]$, $R_1=R\cap U_1[i]$, $\ldots$, $R_{B-1}=R\cap U_{B-1}[i]$ according to the digit in the $i$'s position. Then we iterate through the $B$ sets $R_0, R_1, R_2,\ldots, R_{B-1}$. All nodes of $R_0$ remain in $R$, we remove the nodes of $R_1$ (from $R_1$ and also from $R$) that have a node in distance less than $\alpha$ in $R_0$. Then all nodes of $R_2$ that still (i.e., after we removed some nodes of $R_1$) have a node in distance less than $\alpha$ in $R_0\cup R_1$ are removed from $R_2$ and so on $\ldots$

The runtime  scales linear in $B$ while the final domination scales as $\log_B n$ with $B$. 
Furthermore, instead of using the id-space to split the node sets we can also use a coloring of $G^{\alpha-1}$.
  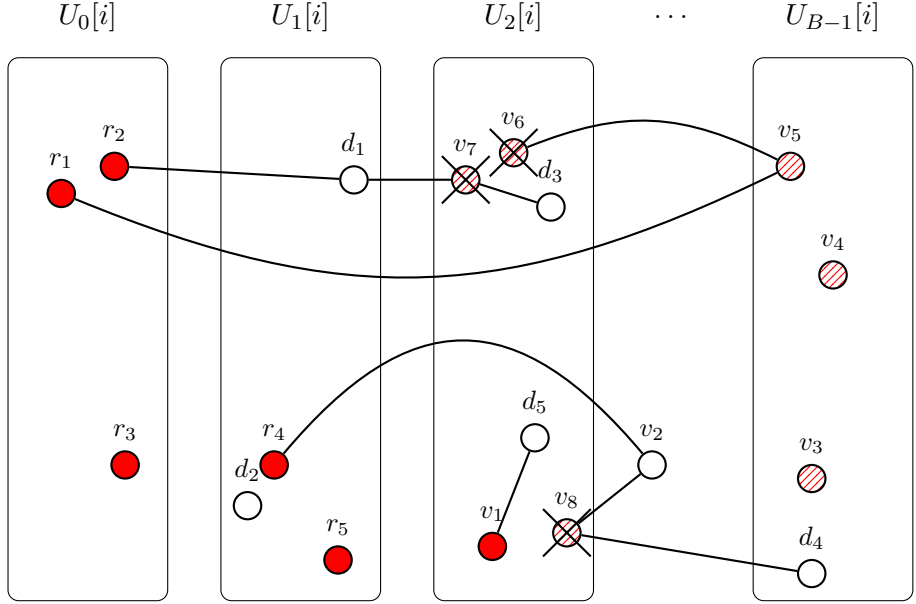
\begin{figure}[H]
    \centering
      \begin{tikzpicture}[xscale=0.7,yscale=1.8]
          \tikzstyle{rulingpath}=[thick]
          \tikzstyle{unconsiderednode}=[draw,circle,thick, pattern=north east lines, pattern color=red]
          \tikzstyle{rulingnode}=[draw,circle,thick,fill=red]
          \tikzstyle{removednode}=[draw,circle,thick]
          \tikzstyle{removesign}=[minimum size=.6cm, draw,thick, cross out]
          \node (U0) at (1.5,8.3) {$U_0[i]$};
          \draw[rounded corners] (0,4) rectangle ++(3,4);

          \node (U1) at (5.5,8.3) {$U_1[i]$};
          \draw[rounded corners] (4,4) rectangle ++(3,4);

          \node (U2) at (9.5,8.3) {$U_2[i]$};
          \draw[rounded corners] (8,4) rectangle ++(3,4);

          \node (Ubminus1) at (15.5,8.3) {$U_{B-1}[i]$};
          \draw[rounded corners] (14,4) rectangle ++(3,4);

          \node (Udots) at (12.5,8.3) {\dots};

          \node[rulingnode, label={\small $r_1$}] (r01) at (1,7) {};
          \node[rulingnode, label={\small $r_2$}] (r02) at (2,7.2) {};
          \node[rulingnode, label={\small $r_3$}] (r03) at (2.2,5.0) {};
          \node[rulingnode, label={\small $r_4$}] (r14) at (5,5) {};
          \node[rulingnode, label={\small $r_5$}] (r15) at (6.2,4.3) {};

          \node[removednode, label={\small $d_1$}] (d11) at (6.5,7.1) {};
          \node[removednode, label={\small $d_2$}] (d12) at (4.5,4.7) {};
					\node[removednode, label={\small $d_3$}] (d13) at (10.2,6.9) {};
					\node[removednode, label={\small $d_4$}] (d14) at (15.1,4.2) {};
					\node[removednode, label={\small $d_5$}] (d15) at (9.9,5.2) {};
          \node[unconsiderednode, label={\small $v_6$}] (v6) at (9.5,7.3) {};
          \node[unconsiderednode, label={\small $v_7$}] (v7) at (8.6,7.1) {};
          \node[unconsiderednode, label={\small $v_8$}] (v8) at (10.5,4.5) {};

          \node[removesign] (remove1) at (v6) {};
          \node[removesign] (remove2) at (v7) {};
          \node[removesign] (remove3) at (v8) {};

          \node[rulingnode, label={\small $v_1$}] (v1) at (9.1,4.4) {};
          
          \node[removednode, label={\small $v_2$}] (v2) at (12.1,5) {};
          \node[unconsiderednode, label={\small $v_3$}] (v3) at (15.1,4.9) {};
          \node[unconsiderednode, label={\small $v_4$}] (v4) at (15.5,6.4) {};
          \node[unconsiderednode, label={\small $v_5$}] (v5) at (14.7,7.2) {};
          \path[rulingpath] (v1) edge (d15);

          \path[rulingpath] (v7) edge (d13);
					\path[rulingpath] (v8) edge (d14);
          
					\path[rulingpath] (r02) edge (d11);
          \path[rulingpath] (d11) edge (v7);

          \path[rulingpath, bend left=25] (r14) edge (v2);
          \path[rulingpath] (v2)  edge (v8);

          \path[rulingpath, bend right=10] (r01) edge (v5);
          \path[rulingpath, bend left=10] (v6) edge (v5);
      \end{tikzpicture}
      \caption{The figure illustrates how bit $i$ is made $3$-independent in  \Cref{alg:krulingwithcolorsgiven}. The figure display some exemplary cases to illustrate the properties of the algorithm. In the presented situation $b=0$ and $b=1$ are already processed. Red nodes represent nodes that remain in $R$ throughout the process of making bit $i$ independent. Nodes that are shaded with red illustrate nodes that still need to be processed.\newline
			$v_6,v_7$ and $v_8$ leave $R$ as they have a neighbor in distance less than three in the set  $R\cap (U_0[i]\cup U_1[i])$.
			The figure illustrates that this distance is measured in $G$ as it does not care whether the red neighbor in distance two or less is reached through  nodes in $U_0[i]\cup U_1[i]$ or nodes in $U_2[i]\cup U_3[i]\cup\ldots \cup U_{B-1}[i]$.\newline
			Node $v_1$ remains in $R$ and $d_5$ is still dominated by it.\newline
            Nodes $d_3$ and $d_4$ are examples for nodes that lost the nodes that dominated them. However, now $d_3$ is dominated by $r_2$ with a larger distance and $d_4$ is dominated by $r_4$ with a larger distance. }\label{fig:examplerulingset}
  \end{figure}

\begin{algorithm}
  \caption{Ruling sets of Simple Graphs}\label{alg:krulingwithcolorsgiven}
  \begin{algorithmic}[1]
    \Require independence parameter $\alpha$, $C$-coloring of $G^{\alpha-1}$, scaling parameter $B \geq 2$
    \Ensure $\big(\alpha, (\alpha-1) \ceil{\log_B(C)})$-ruling set $R$ of $G$
    \Var $U_b[i]$ is the set of vertices with digit $b$ at position $i$ of the $B$-ary representation of their color.

        \State $R = V$
        \For{$i=1,2,\dots, \ceil{\log_B(C)} $} 
					\For{$b=1,2,\ldots, B-1$}
            \For{ \textbf{each} $v \in R\cap U_b[i]$ \textbf{ in parallel}}
                \If{\text{$\exists u\in R\cap \big(U_0[i]\cup \ldots \cup U_{b-1}[i]\big)$ with $\dist_G(u,v)<\alpha$}} \label{alg:crucialLine}
								\State $v$ leaves $R$.
								\EndIf
            \EndFor
        \EndFor
				\EndFor
    \State \Return{$R$}
  \end{algorithmic}
\end{algorithm}
\begin{lemma}\label{thm:finalsimple}
    Let $\alpha$ be a positive integer. For any $B\geq 2$ there exists a deterministic distributed algorithm that, given a $C$-coloring of $G^{\alpha-1}$, computes a $\big(\alpha, (\alpha-1) \ceil{\log_B(C)}\big)$-ruling set of $G$ in $O(B\cdot \alpha \cdot \log_B{C})$ rounds in the $\congestmodel$ model.
\end{lemma}

\begin{proof}
    Assume we are given the graph $G = (V, E)$ and a $C$-coloring of $G^{\alpha-1}$. 
   To show that \Cref{alg:krulingwithcolorsgiven} runs in the \congestmodel model we only need to show how to execute \Cref{alg:crucialLine}, \Cref{alg:krulingwithcolorsgiven} (in $O(\alpha)$ rounds) in the \congestmodel model. We introduce a distance variable $d(v)$ for vertex $v \in V$. All nodes in $R\cap \big(U_0[i]\cup \ldots \cup U_{b-1}[i]\big) $ initialize $d(v)=0$ and all other nodes set $d(v)=\infty$. Then, in $\alpha-1$ iterations each node sends its $d(v)$ value to all of its neighbors and thereafter sets  $d(v) = \min ~ \{d(v)\} \cup \{d(u) + 1 \mid u \in N^1(v) \}$. The runtime for one execution of \Cref{alg:crucialLine} is $O(\alpha)$ which implies that the total runtime is  $O(\alpha \cdot B\cdot \log_B{C})$.
	\myparagraph{$R$ is independent in $G^{\alpha-1}$.} Assume that there are two vertices $u\neq v \in R$ with $\dist(u,v) < \alpha$. Let $i$ be the lowest digit in which their representation differs which exists as the coloring is unique up to distance $\alpha-1$. Without loss of generality let $u$ have the lower value $l$ in the $i$'s position and $v$ the higher $h$. Then in iteration $i$ the node $v$ has been removed from $R$ as $u\in R\cap U_l[i]$ and $v\in R\cap U_h[i]$, a contradiction.

    \myparagraph{$R$ is at least $(\alpha-1) \cdot  \ceil{\log_B(C)}$ dominating.} Before the first iteration we have $R=V$ and $R$ dominates all nodes with distance zero. To prove the claim we show that the domination parameter increases by at most an additive $\alpha-1$ in each iteration.
		So, assume that before some iteration the domination is $\beta$, let $R'$ denote the set $R$ before the iteration, $R$ the set after the iteration and let $v\in V$ be an arbitrary node. Then there exists a node $u\in R'$ with $\dist_G(u,v)\leq \beta$. If $u\in R$ the node $v$ is dominated with distance $\beta\leq \beta+\alpha-1$. If $u\notin R$ there is a node $w\in R$ with $\dist_G(u,w)\leq \alpha -1$. Thus $\dist_G(v,w)\leq \dist_G(v,u)+\dist_G(u,w)\leq \beta+\alpha -1$.
\end{proof}
Unique \id s are a valid coloring of $G^{\alpha-1}$ for any $\alpha$ which implies \Cref{thm:mainSimpleGraph}.
\mainSimpleGraph*
\mainSimpleGraphDelta*
\begin{proof}
  First compute a $c\Delta^2$-coloring of $G$ with Linial's algorithm \cite{linial1987distributive}.
    Thereafter use the $\big(2, \ceil{\log_B c\cdot \degree^2}\big)$-ruling set algorithm from \Cref{thm:mainSimpleGraph} with $B = c\cdot \degree^{2/\beta}$ to compute a $\big(2, \beta \big)$-ruling set.
    The runtime is $O(\degree^{2/\beta} \cdot \log_{c\degree^{2/\beta}} c\Delta^2 +\logstar n)  = O(\degree^{2/\beta} \cdot \beta +\logstar n)$~.
\end{proof}



\clearpage
\bibliographystyle{alpha}
\bibliography{references}

\end{document}